\documentclass[11pt]{article}
\usepackage[margin=1.3in]{geometry}

\usepackage{times}
\usepackage{soul}
\usepackage{url}
\usepackage[hidelinks]{hyperref}
\usepackage[utf8]{inputenc}
\usepackage[small]{caption}
\usepackage{graphicx}
\usepackage{amsmath,amsfonts,mathtools}
\usepackage{booktabs}
\usepackage{algorithm}
\usepackage{algorithmic}
\usepackage{paralist}
\usepackage{color}

\newtheorem{theorem}{Theorem}[section]
\newtheorem{example}{Example}[section]
\newtheorem{definition}{Definition}[section]
\newtheorem{lemma}{Lemma}[section]

\newcommand{\cP}{\mathcal{P}}
\newcommand{\cG}{\mathcal{G}}
\newcommand{\MMS}{\text{MMS}}

\newcommand{\qed}{\unskip\hspace*{1em}\hspace{\fill}$\Box$}
\newenvironment{proof}[1][Proof]{\begin{trivlist}
  \item[\hskip \labelsep {\it #1:}]}{%
    \qed\end{trivlist}}

\title{Maximin-Aware Allocations of Indivisible Goods}

\author{Hau Chan$^1$ \hspace{30pt} Jing Chen$^{2}$ \hspace{30pt} Bo Li$^{2}$ \hspace{30pt} Xiaowei Wu$^{3}$\\
$^{1}$Department of Computer Science and Engineering, University of Nebraska-Lincoln, USA\\
\texttt{hchan3@unl.edu}\\
$^{2}$Department of Computer Science, Stony Brook University, USA\\
\texttt{\{jingchen, boli2\}@cs.stonybrook.edu}\\
$^{3}$Faculty of Computer Science, University of Vienna, Austria\\
\texttt{xiaowei.wu@univie.ac.at}}

\begin{document}

\maketitle

\begin{abstract}
We study envy-free allocations of indivisible goods to agents in settings 
where each agent is unaware of the goods allocated to other agents. 
In particular, we propose the maximin aware (MMA) fairness measure, which guarantees that every agent, given the bundle allocated to her, 
is aware that she does not envy at least one other agent, even if she does not know how the other goods are distributed among other agents.
We also introduce two of its relaxations, and discuss their egalitarian guarantee and existence.
Finally, we present a polynomial-time algorithm, which computes an allocation that approximately satisfies MMA or its relaxations.
Interestingly, the returned allocation is also $\frac{1}{2}$-approximate EFX when all agents have subadditive valuations,
which improves the algorithm in [Plaut and Roughgarden, SODA 2018].
\end{abstract}

\section{Introduction}
In the last few years or so, there has been a tremendous demand for
fair division services to provide systematic and \emph{fair} ways of
dividing a set of (indivisible) goods such as tasks, courses, and properties among a group of agents 
so that the agents do not envy each other. 
Such demand gave rise to, for examples, Spliddit\footnote{http://www.spliddit.org/},
the University of Pennsylvania's Course Match\footnote{https://mba-inside.wharton.upenn.edu/course-match/},
and Fair Outcomes, Inc\footnote{https://www.fairoutcomes.com/}
which are all based on mathematical and fairness notions. 
To capture the fairness of an allocation, 
which is arguably initiated by the work of \cite{foley1967resource},
{\em envy-freeness (EF)} (and its relaxations, such as EF1 and EFX) is often used to ensure that
each agent should not envy or prefer the allocated goods of other agents. 


In this paper, we study an envy-free allocation domain 
where the planner of the division tasks
wishes to withhold allocation information of others from the user 
or the user simply does not know the allocation of others in the system. 
There are a couple of good reasons why it is desirable 
for the planner to withhold such information. 
First, in many private fair allocations of goods such as tasks or gifts, 
the planner requires the system to preserve anonymity 
as not to give away the received bundles of other agents. 
Second, due to the large number of (unrelated) agents and items 
that could be potentially be involved in the division tasks (e.g., on the Internet such as MTurks), 
it is not meaningful for the planner to provide 
such information due to various reasons.
Motivated by this domain, we focus on answering the following questions. 

\emph{When indivisible goods are to be allocated among unaware agents,
what is the appropriate envy-free notion and 
how efficiently can the allocation be found subject to the notion?}

Proportionality (PROP) and maximin share (MMS) \cite{budish2011combinatorial} 
are two widely studied and well accepted fair allocation notions,
both of which are defined for unaware agents. 
In PROP, it is required that the value of every agent's bundle is at least a $\frac{1}{n}$ 
fraction of her value for the whole goods, where $n$ is the number of agents.
It is well known that such an allocation may not exist for indivisible goods,
thus a weaker and more realistic notion is desired for indivisible goods.
MMS is a proper relaxation of PROP, which studies an adversarial situation: 
when the goods are partitioned into bundles and an agent would always get
the least preferred bundle of goods, what is the best way she can partition the goods. 
The value of such a bundle is the MMS value of the agent. 
In addition to its non-existence result, MMS allocation only guarantees each agent's best minimum value, 
and the value of some agent's bundle can still be the least compared with others, 
which may cause significant envy (demonstrated by the below example). 

\begin{example}\label{ex:intro:mms}
	Suppose there are $2$ agents and $2$ goods,  
	$v_{1}(\{1\}) = H$, $v_{1}(\{2\}) = 1$ and $v_{2}(\{1\}) = 1$, $v_{2}(\{2\}) = H$,
	where $H$ is a sufficiently large number and $v_1$ and $v_2$ are valuation functions. 
	Then each agent's MMS value is 1 and $A_{1}=\{2\}$, $A_{2}=\{1\}$ is an MMS allocation, 
	but each agent envies each other. 
	However, if we exchange their allocations, everyone's value can be improved to $H$. 
\end{example}

Recently, \cite{aziz2018knowledge} introduced epistemic envy-free (EEF) notion to study unaware setting.
With respect to EEF, each agent is satisfied if 
there exists one reallocation of the goods that she does not get among the other agents,
such that her value for her bundle is at least as good as every bundle in this reallocation. 
This measure is not robust if the reallocations are restricted by agent's reasoning (e.g., adversarial settings in MMS).
Moreover, it can be shown that EEF and PROP allocations (and their relaxations) barely exist and cannot be properly approximated\footnote{Consider three agents with two goods such that every agent has
value 1 for each good. Then the agent that receives no good has no
bounded guarantee for both PROP and EEF (and their relaxations by removing any item from the items she has not obtained).}.

\paragraph{Our Contributions.}
In this paper, we focus on modeling the envy-freeness for indivisible goods allocation
as well as deriving new algorithms to find (approximately) fair allocations
subject to different fairness notions in an unaware environment. 
Our main contributions include the following.

First, we introduce a novel fairness notion of maximin aware (MMA), which guarantees that the agent's bundle value is at least as much as her value for some other agent's bundle, no matter how the remaining goods are distributed, i.e., there is always somebody who gets no more than her.
MMA combines the notion of epistemic envy-freeness \cite{aziz2018knowledge} and MMS \cite{budish2011combinatorial} where each agent may not know or care about the exact allocation of the remaining goods to other agents, but can still guarantee that the value for her own bundle is at least as much as she can obtain if she is given the chance to repartition the remaining goods and receive the least valued portion. We provide a detailed picture of the relationship between MMA and classic fairness notions for different valuations.
	
Then we show MMA is a strong requirement and cannot be guaranteed in general, and provide two relaxations of MMA: MMA1 and MMAX. We also prove that MMA1 (and MMAX) potentially has stronger egalitarian guarantee than EF1 and such an allocation is guaranteed to exist for the following situations: (1) there are at most three agents with additive valuations and (2) there are any number of agents but all of them have identical submodular valuation.
If the above requirement of submodularity is replaced by strictly increasing subadditivity, an MMAX allocation is guaranteed to exist.
In contrast, MMS allocation may not exist even for three agents with additive valuations~\cite{kurokawa2018fair} 
and an EFX allocation is only known to exist when there are two agents \cite{plaut2018almost}.
	
We present a polynomial-time algorithm that computes an allocation such that every agent is either $\frac{1}{2}$-approximate MMA or exactly MMAX for additive valuations. For several specific classes of valuations, such as binary additive or additive with identical preference ranking, our algorithm returns an exact MMAX and EFX allocation. It is shown in \cite{plaut2018almost} that a $\frac{1}{2}$-approximate EFX allocation exists for general subadditive valuations, but finding it may need exponential time, which leaves an open question whether such an allocation can be found in polynomial time. We show that the allocation by our algorithm computes is also $\frac{1}{2}$-approximate EFX when all agents have subadditive valuations, and thus answer this problem affirmatively.

\paragraph{Our Approaches.}
To show the existence of MMA1 allocations, we first show that if all $n$ agents have identical submodular valuations, 
a leximin $n$-partition of the goods is MMA1.
Next, we provide a divide-and-choose style algorithm for 3 agents with additive valuations (need not to be identical).
Basically, we let the first agent make a leximin $3$-partition, and other 2 agents select.
Then no matter which subset is left for the first agent, she is satisfied with respect to MMA1.
However, the difficulty comes from how the other 2 agents should select.
By carefully examining all possible cases, we show that there is a way to select and redivide the three subsets such that both of them are satisfied with respect to MMA1.

To design an efficient algorithm for MMAX or EFX allocations, 
we construct a bipartite graph between agents and goods.
We show that there is a way to use a maximum weighted matching between the unenvied agents and the remaining goods to guide our allocation.
Eventually, we show that all the goods are allocated, and the final allocation is either $\frac{1}{2}$-approximate MMA or exact MMAX for additive valuations and $\frac{1}{2}$-approximate EFX for subadditive valuations.

\paragraph{Other Related Works.}
For indivisible goods allocations, EF or PROP allocations may not always exist\footnote{For example, there are two agents but only one good.}.
\cite{budish2011combinatorial} introduced the relaxed concept of envy-freeness up to one good (EF1).
In an EF1 allocation, it is only required that each agent's value for a bundle is at least as much as her value for every other agent's bundle minus a single good (in the bundle).
It is shown in \cite{lipton2004approximately} that an EF1 allocation always exists, and can be found in polynomial time.
In \cite{caragiannis2016unreasonable}, the authors introduced a strictly stronger fairness notation than EF1 called envy-free up to any good (EFX), where the comparison is made to ``any'' single good instead of ``a'' single good.
The state-of-the-art results show that an EFX allocation exists in the following settings:
(1) there are 2 agents, or
(2) there are any number of agents but all of them have the identical valuation \cite{plaut2018almost}.
It is still an open question whether an EFX allocation exists in general, even for additive valuations.

Recently, we see many new fairness notions adapted for different settings,
such as Nash social welfare \cite{caragiannis2016unreasonable},
epistemic envy-freeness \cite{aziz2018knowledge},
pairwise maximin share \cite{caragiannis2016unreasonable},
and groupwise maximin share \cite{barman2018groupwise}.
There is also a line of works making connections (or constraints) of fair allocations
to graphs or networks~\cite{abebe2017fair,bei2017networked,aziz2018knowledge}.
More works on the fair allocation problem under different constraints can be found in~\cite{ferraioli2014regular,bouveret2017fair,biswas2018fair}.

\paragraph{Organization.}
We provide necessary definitions and introduce MMA in Section \ref{sec:def}.  
We show its connections with classic fairness notions for different classes of valuations in Sections \ref{sec:MMA}.
In Sections \ref{sec:MMA1} and \ref{sec:MMAx}, we study the two relaxations of MMA, MMA1 and MMAX,
and discuss their connections with other fairness notions and their existences in different settings.
Finally, in Section \ref{sec:alg}, we present a polynomial time algorithm to compute a $\frac{1}{2}$-approximate MMAX allocation for additive valuations;
and $\frac{1}{2}$-approximate EFX allocation for subadditive valuations.

\section{Preliminaries and Definitions}\label{sec:def}

In the envy-free division problem on indivisible goods,
there is a set $N = \{1, ..., n\}$ of $n$ agents
and a set $M = \{1, ..., m\}$ of $m$ indivisible goods.
Each agent $i \in N$ has an valuation function $v_i: 2^{N} \to \mathbb{R}_{\ge 0}$
that maps every subset of goods to a non-negative real number.
We assume that $v_i$ is normalized (i.e., $v_i(\emptyset)=0$)
and montone (i.e., $v_i(S) \le v_i(T)$ for every $S \subseteq T \subseteq M$)
for each agent $i \in N$.
For convenience we use $v(j)$ to denote $v(\{j\})$, for every valuation $v$ and every good $j\in M$.
Throughout the whole paper, we call an valuation $v$
\begin{enumerate}
	\item \emph{additive} if $v(S) = \sum_{j\in S}v(j)$ for each $S \subseteq M$;
	\item \emph{binary additive} (BA) if $v$ is additive, and $v(j)\in \{0,1\}$ for any good $j$;
	\item \emph{subadditive} (SA) if $v(S \cup T) \le v(S) + v(T)$ for any $S, T \subseteq M$;
	\item \emph{submodular} (SM) if for any $S\subseteq T$ and $e\in M\setminus T$,
	$v(S\cup\{e\}) - v(S) \geq v(T\cup\{e\}) - v(T)$.
\end{enumerate}

An \emph{allocation} $A=(A_1, ..., A_n) = (A_i)_{i \in N}$ is a partition of $M$ among agents in $N$, i.e., $\bigcup_{i\in N} A_i = M$ and $A_i \cap A_j = \emptyset$ for any $i \neq j$.
In other words, $A_i$ is the set of goods allocated to agent $i$.
Let $A_{-i} = \cup_{j\neq i} A_j$ be the goods not assigned to agent $i$.
Without loss of generality, we assume that for every $j \in M$,
there exists $i \in N$ such that $v_i(j) > 0$.

Below we state some standard fairness definitions.

\begin{definition}[EF]
	For any $\alpha\in [0,1]$,
	an allocation $A$ is $\alpha$-envy-free ($\alpha$-EF) if for any $i, j \in N$,
	$v_{i}(A_{i})\geq \alpha\cdot v_{i}(A_{j})$.
	The allocation is EF when $\alpha=1$.
\end{definition}

\begin{definition}[EF1]
	For any $\alpha\in [0,1]$,
	an allocation $A$ is $\alpha$-EF1 if for any $i, j \in N$,
	there exists $e\in A_{j}$ such that
	$v_{i}(A_{i})\geq \alpha\cdot v_{i}(A_{j}\backslash \{e\})$.
	The allocation is EF1 when $\alpha=1$. 
\end{definition}

\begin{definition}[EFX]
	For any $\alpha\in [0,1]$,
	an allocation $A$ is $\alpha$-EFX if for any $i, j \in N$,
	$v_{i}(A_{i})\geq \alpha\cdot v_{i}(A_{j}\backslash \{e\})$ for any $e\in A_{j}$.
	The allocation is EFX when $\alpha=1$.
\end{definition}

Next, we introduce another fairness notion called \emph{maximin share} (MMS).
Let $\Pi_{n}(S)$ be the set of all $n$-partitions of a set $S$.
The maximin share of agent $i$ on $S$ among $n$ players is
\begin{equation*}
\MMS_i(S,n) = \max_{\pi\in \Pi_{n}(S)} \min_{S\in \pi} v_{i}(S).
\end{equation*}

\begin{definition}[MMS]
	For any $\alpha \in [0,1]$,
	an allocation $A$ is $\alpha$-MMS if for any $i \in N$,
	$v_{i}(A_{i})\geq \alpha\cdot \MMS_i(M,n)$.
	The allocation is MMS when $\alpha=1$.
\end{definition}

Next we define the epistemic envy-free (EEF) 
\cite{aziz2018knowledge}, which is the most relevant fairness notion to this work. 

\begin{definition}[EEF] \label{def_eef}
	An allocation $A$
	is EEF if for any $i$, there exists an allocation 
	$\bar{A}^i=(\bar{A}^i_{j})_{j\in N}$ such that 
	$\bar{A}^i_i = A_i$ and
	$v_{i}(A_{i})\geq v_{i}(\bar{A}^i_{j})$ for any $j\in N$.
\end{definition}

For the purpose of presenting our results, we provide the definition of proportionality fairness.

\begin{definition}[PROP]
	An allocation $A$ is PROP if for any $i \in N$,
	$v_{i}(A_{i})\geq \frac{1}{n}\cdot v_i(M)$.
\end{definition}

Finally, we say that $X\xRightarrow{type}Y$ if every $X$ allocation is also a $Y$ allocation when agents have $type$ valuations, where $X$ and $Y$ are the fairness notions (e.g., MMA1 and MMS) and $type$ is the function type (e.g. SA and SM).

\section{Maximin-Aware Allocation} \label{sec:MMA}

In this section, we introduce the fairness notion of maximin-aware (MMA), 
and study its connection to other existing fairness notations.

\begin{definition}[MMA]
	For any $\alpha \in [0,1]$,
	an allocation $A$ is $\alpha$-MMA if for any agent $i \in N$, $v_{i}(A_{i})\geq \alpha \cdot\MMS_{i}(A_{-i}, n-1)$.
	The allocation is MMA when $\alpha=1$.
\end{definition}

There is a modeling and conceptual advantages of MMA:
an MMA allocation guarantees for every agent,
no matter how the goods that she does not get are distributed
among the other agents, there is always an agent who
gets no more than her.
Compared to EEF, where the happiness of each each agent depends on the existence of a reallocation of the other goods,
MMA provides a robust way for the agents to reason about their reallocations of the remaining goods. 

Observe that an MMS allocation need not be MMA.
Recall, Example \ref{ex:intro:mms} where $A_{1}=\{2\}$ and $A_{2}=\{1\}$ is an MMS allocation.
However, this allocation is far from being MMA since $v_{1}(A_{1})\ll v_{1}(A_{2})$ and $v_{2}(A_{2})\ll v_{2}(A_{1})$.
The only MMA allocation in this example is $A_{1}=\{1\}$ and $A_{2}=\{2\}$.

Seemingly MMA being a stronger definition than MMS,
Example \ref{MMAnotMMS} shows that this is not true.

\begin{example}\label{MMAnotMMS}
	Suppose there are four agents ($N=\{1,2,3,4\}$) and eight goods ($M=\{1,2,\cdots,8\}$).
	Agent 1 has the additive valuation shown in Table \ref{table:example:MMA:mms}.\footnote{We do not explicitly describe the valuations for the other agents since it is easy to design other agents' valuations such that $S$ is allocated to agent 1.}
	It is easy to check that $\MMS_{1}(M,4)=0.5$.
	Let $S=\{5,6,7,8\}$. Then $v_1(S)=0.4<\MMS_{1}(M,4)$.
	That is, for agent 1, $S$ does not satisfy the condition of MMS .
	However, $\MMS_1(M\setminus S,3)=0.4=v_1(S)$, which satisfies the condition of MMA.
\begin{table}[ht]
	\begin{center}
		\begin{tabular}{|c|ccc|}
			\hline
			\mbox{Goods} & $1$ & $2,3,4$ & $5,6,7,8$ \\
			\hline
			\mbox{Value} &1 & 0.4 & 0.1\\
			\hline
		\end{tabular}
	\end{center}
	\caption{MMA does not imply MMS in general.}
	\label{table:example:MMA:mms}
\end{table}
\end{example}

However, when all agents have binary additive (BA) valuations, MMA actually implies MMS.

\begin{lemma} \label{lem:MMAtoMMS}
	MMA $\xRightarrow{BA}$ MMS.
\end{lemma}
\begin{proof}
	Fix any agent $i \in N $, let $m_i = v_i(M)$ be the number of goods agent $i$ is interested in.
	Then $\MMS_i(M,n) = k_i = \lfloor \frac{m_i}{n} \rfloor$.
	Fix any MMA allocation $A$, we show that it is also MMS.
	Suppose otherwise, i.e., $v_i(A_i) < k_i$.
	Then we have
	\begin{equation*}
	\MMS_{i}(A_{-i}, n-1) = \lfloor \frac{m_i-v_i(A_i)}{n-1}\rfloor
	\geq \lfloor \frac{m_i-\frac{m_i}{n}}{n-1}\rfloor \geq k_{i},
	\end{equation*}
	which contradicts the fact that $A$ is MMA.
	Hence $A$ is also an MMS allocation.
\end{proof}

Indeed, similar to MMS, MMA is a more realistic and appropriate definition for indivisible goods allocation, in the sense that MMA is more approachable than EF and PROP for additive valuations.
In the following, we show that PROP $\xRightarrow{A}$ MMA.
Combining with Theorem 4 of~\cite{aziz2018knowledge}, it implies EF $\xRightarrow{A}$ EEF $\xRightarrow{A}$ PROP $\xRightarrow{A}$ MMA.

\begin{lemma} \label{lem:PROPtoMMA}
EF $\xRightarrow{A}$ EEF $\xRightarrow{A}$ PROP $\xRightarrow{A}$ MMA
\end{lemma}
\begin{proof}
	Fix any PROP allocation $A$ and any agent $i$.
	It follows that $v_i(A_i) \ge \frac{1}{n}\cdot v_i(M)$.
	Since the valuations are additive,
	\begin{equation*}
	\MMS_i(M\setminus A_i, n-1) \leq \frac{1}{n-1}\cdot (1-\frac{1}{n})\cdot v_i(M) = \frac{v_i(M)}{n}.
	\end{equation*}
	Hence, the allocation is also MMA.
\end{proof}

Note that the other direction is not true:
in Example~\ref{MMAnotMMS}, agent~1 is satisfied with $S=\{5,6,7,8\}$ with respect to MMA,
but the portion is only $\frac{0.4}{2.6} < 0.25$.
On the other hand, MMA is still a strong requirement for indivisible goods allocation,
in the sense that it cannot be satisfied even in some simple settings.

\begin{lemma} \label{lem:noMMA}
	MMA allocations may not exist even when the agents have identical BA valuation.
\end{lemma}

\begin{proof}
	Assume there are $n$ agents and $m=kn-1$ goods, where $k>1$.
	Each agent's value on every good is $1$.
	For any agent $i$, if she gets $|A_{i}| \leq k-1$ goods then
	\begin{align*}
	\MMS_{i}(A_{-i}, n-1) = & \lfloor \frac{kn-1-|A_{i}|}{n-1}\rfloor \\
	= & \lfloor k + \frac{k-1-|A_{i}|}{n-1}\rfloor \geq k > |A_{i}|.
	\end{align*}
	Thus in any MMA allocation, we have $|A_{i}| \geq k$ for every agent $i$, which is impossible with $kn-1$ goods in total.
\end{proof}

In the following two sections, we introduce two relaxations of MMA based on EF1 and EFX, respectively.

\section{Maximin-Aware up to One Good} \label{sec:MMA1}

In this section, we relax the MMA notion to maximin-aware up to one good (MMA1) that is analogous to EF1.

\begin{definition}[MMA1]\label{def:MMA1}
	Fix any $\alpha \in[0,1]$, an allocation $A$ is
	$\alpha$-MMA1 if for any $i$, there exists $e\in A_{-i}$ such that
	\begin{equation*}
	v_{i}(A_{i})\geq \alpha\cdot \MMS_{i}(A_{-i}\backslash \{e\}, n-1).
	\end{equation*}
	The allocation is MMA1 when $\alpha=1$.
\end{definition}

\subsection{Connection to EF1 and MMS}

We first note that while MMA1 is a relaxed version of MMA, it may have better egalitarian guarantee than EF1 allocations.
An MMA1 allocation $A$ guarantees that for each agent $i$ and her favorite item $e\in A_{-i}$ (assume $e\in A_{k}$),
$v_{i}(A_{i})$ is at least as large as the worst bundle in {\em any} $(n-1)$-partition of $A_{-i}\setminus\{e\}$.
However, EF1 implies that there {\em exists} an $(n-1)$-partition of $A_{-i}\setminus\{e\}$, 
i.e., $A_{i'}$ for $i'\neq k$ and $A_{k}\setminus\{e\}$, such that $v_{i}(A_{i})$ is at least as large as the worst bundle, i.e. $A_{k}\setminus\{e\}$.
We illustrate this as the following example.

\begin{example}	\label{example:EF1}
	Consider $n+1$ agents and $2n+1$ goods, where every agent has identical additive valuation $v$ with
	$v(\{j\}) = n$ for $1\leq j\leq n$ and $v(\{j\}) = 1$ for $n+1\leq j\leq 2n+1$.
	Thus allocation $A_{i}=\{i,n+i\}$ for $1\leq i\leq n$ and $A_{n+1}=\{2n+1\}$ is EF1.
	In such an allocation, agent $n+1$ only gets value 1.
	However, there is an allocation (for example: $A_{i}=\{i\}$ for $1\leq i\leq n$ and $A_{n+1}=\{n+1,\cdots,2n+1\}$) 
	where everyone obtains value at least $n$.\footnote{We note that there is no exact EF allocation for this example.}
\end{example}

In Example \ref{example:EF1}, an EF1 allocation can be arbitrarily bad to agent $n + 1$.
However, we show that any MMA1 allocation $A$ in this example guarantees that each agent's value is at least $\lceil\frac{n}{2}\rceil$.
If everyone's value is at least $n$, then our claim trivially holds.
Suppose for some agent $i$, $v(A_{i})=k<n$.
Then among all goods in $A_{-i}$, there are $n$ goods with value $n$ and $n+1-k$ goods with value 1.
Thus by removing the most favorite good, the maximin share for the remaining set is $\min\{n, n+1-k\}$.
Since $A$ is an MMA1 allocation, $v(A_{i})=k\geq n+1-k$, which means $v(A_{i}) \geq \lceil\frac{n}{2}\rceil$.

In the following, we show that MMS implies MMA1 when the agents's valuations are submodular.

\begin{lemma} \label{thm:MMA1:mms:add}
	MMS $\xRightarrow{SM}$ MMA1.
\end{lemma}
\begin{proof}
	Let $A$ be an MMS allocation.
	It follows that $v_i(A_i) \geq \MMS_i(M, n)$ for every agent $i$.
	Fix any agent $i \in N$, it suffices to show that there exists $e\in A_{-i}$ such that $v_i(A_i) \geq \MMS(A_{-i}\setminus\{e\}, n-1)$.
	
	If for every $e\in A_{-1}$, $v_{i}(A_i\cup\{e\})=v_{i}(A_i)$, then by the submodularity of $v$, we have
	\begin{equation*}
	v_i(M) \leq v_i(A_i) + \sum_{e\in A_{-i}}(v_i(A_i\cup \{e\}) - v_i(A_i)) = v_i(A_i),
	\end{equation*}
	where the inequality holds since the marginal contribution of each good does not increase.
	Thus $A$ is also an MMA1 allocation for agent $i$.
	
	Otherwise, let $e\in A_{-i}$ be such that $v_{i}(A_i\cup\{e\}) > v_{i}(A_i)$.
	For the sack of contradiction, suppose there exists a partition $B=(B_j)_{j\in N\setminus \{i\}}$ of $A_{-i} \setminus \{e\}$ such that $v_i(B_j) > v_i(A_i)$ for all $j \in N\setminus \{i\}$.
	Then if we include $e$ into $A_i$, the resulting partition of $M$ (into $n$ parts) has minimum value strictly larger than $\MMS_i(M,n)$ on every agent, which contradicts the definition of $\MMS_i(M,n)$.
\end{proof}

By a similar argument, every $\alpha$-MMS allocation is naturally $\alpha$-MMA1.
Recall Example \ref{MMAnotMMS} where we show an MMA allocation is not necessarily MMS (thus neither an MMA1 allocation), this implication is strict.

%


Next, we show that for binary additive valuations, MMA1 and MMS are actually equivalent.

\begin{lemma} \label{thm:MMA1:mms}
	MMA1 $\xLeftrightarrow{BA}$ MMS.
\end{lemma}
\begin{proof}
	Since binary additive functions are submodular, it suffices to prove the "$\Rightarrow$" direction.
	Fix any MMA1 allocation $A = (A_i)_{i\in N}$, we prove that $v_i(A_i) \geq k_i$ for every agent $i$, where $k_i = \MMS_i(M,n) = \lfloor \frac{v_i(M)}{n} \rfloor$.
	
	Fix any agent $i$ and assume the contrary that $v_i(A_i) < k_i$.
	Let $m_i = v_i(M)$.
	Then for any $e\in A_{-i}$ we have 
	\begin{align*}
	& \MMS_{i}(A_{-i}\setminus \{e\}, n-1) = \lfloor \frac{m_i-1-v_i(A_i)}{n-1}\rfloor \\
	\geq & \lfloor \frac{m_i-1-k_i+1}{n-1}\rfloor \geq k_{i} > v_i(A_i),
	\end{align*}
	which contradicts the fact that $A$ is MMA1.
\end{proof}

Recall that in Lemma~\ref{lem:MMAtoMMS}, we prove that MMA $\xRightarrow{BA}$ MMS, which is actually implied by Lemma~\ref{thm:MMA1:mms}, since MMA $\Rightarrow$ MMA1, for any valuations.

It is shown in \cite{kurokawa2018fair} that even for three agents with additive valuations,
MMS allocations do not always exist.
As we will show later, for a broader class of situations, an MMA1 allocation is guaranteed to exist.

\subsection{The Existence of MMA1 Allocations}\label{sec:MMA1,MMAX:existence}

By Lemma \ref{thm:MMA1:mms}, we know that an MMA1 allocation always exists when the valuations are binary additive.
In this section, we explore the existence of MMA1 allocations for instances with subadditive valuations.

\paragraph{Leximin $n$-partition.}
First, for any valuation $v$ over a set $M$ of goods, we define the {\em leximin $n$-partition} as follows.
A leximin $n$-partition is a partition which divides $M$ into $n$ subsets and
maximizes the lexicographical order when the values of the partitions are sorted in non-decreasing order.
In other words, a leximin $n$-partition maximizes the minimum value over all possible $n$-partitions, and if there is a tie it selects the one maximizing the second minimum value, and so on.

\medskip

In the following, we show that MMA1 allocations exist when all agents have identical submodular valuation, or when there are three agents with (different) submodular valuations.

\paragraph{Identical submodular valuation.}
We first show that when all agents have identical submodular valuation,
the leximin $n$-partition of $M$ is indeed an MMA1 allocation.

\begin{theorem}\label{thm:MMA1:sym:agent}
	When all agents have the identical submodular valuation, the leximin $n$-partition of $M$ is MMA1.
\end{theorem}
\begin{proof}
	Let $v$ be the submodular valuation and $A$ be the leximin $n$-partition of $M$.
	We show that allocating each partition $A_i$ to agent $i$ gives an MMA1 allocation.
	As before, if for all $e\in A_{-i}$, $v(A_i\cup \{e\}) = v(A_i)$, then we have
	\begin{equation*}
	v(A_i) = v(M)\geq \MMS(A_{-i},n-1)
	\end{equation*}
	by submodularity of $v$. Otherwise let $e\in A_{-i}$ be any good such that $v(A_i\cup\{e\}) > v(A_i)$.
	
	Assume the contrary that $v(A_i) < \MMS(A_{-i}\setminus\{e\},n-1)$.
	
	Let $B = (B_j)_{j\in N\setminus\{i\}}$ be the partition of $A_{-i}\setminus\{e\}$ such that $v(A_i) < v(B_j)$ for all $j\neq i$.
	Then by including $e$ into $A_i$, we obtain a partition of $M$ such that every bundle has value strictly larger than $v(A_i)$.
	In other words, the resulting partition is strictly better than the leximin $n$-partition, which is a contradiction.
\end{proof}

\paragraph{Three agents with different additive valuations.}
Next we show a divide-and-choose style algorithm (shown in Algorithm~\ref{alg:3agents}) which gives an MMA1 allocation.

\begin{algorithm}[H]
	\caption{\hspace{-3pt}  Divide-and-Choose Algorithm for Three Agents}
	\label{alg:3agents}
	\begin{enumerate}
		\item Agent 3 divides $M$ as a leximin partition, which is denoted by $A = (A_1, A_2, A_3)$.
		
		\item \label{step:n3:1}If agent 1's and agent 2's favorite bundles are different, each of them takes her favorite bundle and the remaining bundle is allocated to agent 3.
		
		\item \label{step:n3:2}Otherwise (agent 1 and agent 2 have the same favorite bundle), we compare their second favorite bundle.
		If their second favorite bundles are also the same, agent 2 takes her favorite two bundles (w.o.l.g, say $A_1$ and $A_2$), and leaves $A_3$ to agent 3.
		Next, agent 2 repartitions $A_1\cup A_2$ by her leximin $2$-partitions, denoted by $B_{1}$ and $B_{2}$.
		Agent 1 chooses her favorite one in $B_{1}$ and $B_{2}$, and the other one is allocated to agent 2.
		
		\item Finally, suppose their second favorite bundles are different.
		W.o.l.g, suppose $v_{1}(A_{1})\geq v_{1}(A_{2})\geq v_{1}(A_{3})$ and $v_{2}(A_{1}) \geq v_{2}(A_{3})\geq v_{2}(A_{2})$.
		\begin{enumerate}
			\item  \label{step:n3:3}If $2\cdot v_{1}(A_{2}) > v_{1}(A_{1})+v_{1}(A_{3})$ and $2\cdot v_{2}(A_{3}) > v_{2}(A_{1})+v_{2}(A_{2})$, allocate $A_{2}$ to agent 1, $A_{3}$ to agent 2 and $A_1$ to agent 3.
			\item \label{step:n3:4}Otherwise (assume w.l.o.g. that $2\cdot v_{1}(A_{2}) \leq v_{1}(A_{1})+v_{1}(A_{3})$).
			Agent 2 takes $A_{1}$ and $A_{3}$, which are her favorite two bundles, and leave $A_2$ to agent 3.
			Next, agent 2 repartitions $A_{1}\cup A_{3}$ by her leximin $2$-partition $C = (C_1, C_2)$, and let Agent 1 choose her favorite one in $C_{1}$ and $C_{2}$. The other one is allocated to agent 2.
		\end{enumerate}
	\end{enumerate}
\end{algorithm}

\begin{theorem}\label{thm:MMA1:3agents}
	When $n=3$ and the valuations are additive, Algorithm~\ref{alg:3agents} computes an MMA1 allocation.
\end{theorem}
\begin{proof}
	First, observe that since $(A_1,A_2,A_3)$ is a leximin $n$-partition of agent 3, and she always gets one of $A_1$, $A_2$ and $A_3$, by Theorem \ref{thm:MMA1:sym:agent}, she is satisfied with respect to MMA1.
	
	Next, we consider agents 1 and 2.
	
	If the final allocation is from Case~\eqref{step:n3:1}, i.e., both of them get their favorite bundles, then they are satisfied with respect to MMA1, given that the valuations are additive.
	
	If the final allocation is from Case~\eqref{step:n3:2}, i.e., agent 2 repartitions $A_{1}$ and $A_{2}$, the favorite two bundles of both agents, and let agent 1 choose first, then we show that both agents are satisfied with respect to MMA1.
	First, Agent 1 will always be happy, since the bundle she chooses must have value at least $\frac{1}{3}\cdot v_{1}(M)$, by the additivity of $v_1$.
	For agent 2, by leximin partitioning $A_1\cup A_2$ into $B_{1}$ and $B_{2}$ (suppose $v_2(B_1)\leq v_2(B_2)$), we have $v_2(B_2)\geq \frac{1}{3}\cdot v_2(M)$. Thus if $B_2$ is left for agent 2, she will be satisfied.
	Otherwise agent 2 gets $B_1$, which satisfies $v_2(B_1) \geq v_2(A_3)$, since agent 2 is still able to leave $A_1$ and $A_2$ unchanged.
	Observe that there must exist some $e\in B_2$, such that $v_2(B_1\cup\{e\}) > v_2(B_1)$, as otherwise $v_2(B_1) = v_2(B_1\cup B_2)$, which is impossible.
	Moreover, we have $v_{2}(B_{1})\geq v_{2}(B_{2}\backslash \{e\})$, otherwise $B_{1}$ and $B_{2}$ is not a leximin partition.
	Thus $v_{2}(B_{1})\geq \frac{1}{2}\cdot v_{2}(A_3\cup B_2\backslash\{e\})$, which implies that agent 2 is also satisfied with respect to MMA1.
	
	If the final allocation is from Case~\eqref{step:n3:3}, agent 1 and agent 2 get a bundle with value at least $\frac{1}{3} v_1(M)$ and $\frac{1}{3} v_2(M)$, respectively.
	Thus they are both satisfied.
	
	If the final allocation is from Case~\eqref{step:n3:4}, i.e., agent 2 takes away her favorite two bundles $A_{1}$ and $A_{3}$ and leximin-partitions them, then similar to Case \ref{step:n3:2}, no matter which one is left for her, she will be satisfied.
	Note that for agent 1, since $v_{1}(A_{1})+v_{1}(A_{3}) \geq 2\cdot v_{1}(A_{2})$, the bundle she gets has value at least $\frac{1}{2}\cdot v_1(A_1\cup A_3)\geq \frac{1}{3}\cdot v_1(M)$.
	Thus agent 1 is also satisfied.
\end{proof}

\paragraph{Remark.}
We note that when there are two agents and both of them have subadditive valuations (do not need to be identical), the definition of MMA1 degenerates to be EF1.
Thus an MMA1 allocation always exists.
Indeed, using approaches similar to Theorems~\ref{thm:MMA1:sym:agent} and~\ref{thm:MMA1:3agents},
we can show that a leximin allocation is an MMA1 allocation if (1) all agents have identical strictly increasing subadditive valuations, or (2) there are three agents with (different) strictly increasing subadditive valuations.
We say that a subadditive function $v$ is strictly increasing if $v(S\cup \{e\}) > v(S)$ for any $e\notin S$.

\section{Maximin-Aware up to Any Good} \label{sec:MMAx}


First, we define MMAX as follows.

\begin{definition}[MMAX]\label{def:MMAX}
	Fix any $\alpha \in [0,1]$, an allocation $A$ is $\alpha$-MMAX if for any $i$,
	\begin{equation*}
	v_{i}(A_{i})\geq \alpha\cdot \MMS_{i}(A_{-i}\setminus \{e\}, n-1)
	\end{equation*}
	for any $e\in A_{-i}$.
	The allocation is MMAX when $\alpha=1$.
\end{definition}


We first note that MMAX allocations give an even better egalitarian guarantee than MMA1 and EF1 allocations.
Recall Example \ref{example:EF1},
where we show there is an EF1 allocation such that some agent only gets value 1 and
any MMA1 allocation guarantees each agent's value is at least $\lceil\frac{n}{2}\rceil$.
It is not difficult to show that any MMAX allocation guarantees each agent's value is at least $n$.

Obviously, an MMAX allocation is also MMA1.
Together with Lemma \ref{thm:MMA1:mms}, we have
$$\mbox{MMAX $\xRightarrow{BA}$ MMA1 $\xLeftrightarrow{BA}$ MMS}.$$
However, beyond BA valuations, the above implication does not hold.
We can see this by borrowing an example from \cite{kurokawa2018fair}, for which there is no MMS allocation,
but MMAX (and thus MMA1) allocations do exist.

\begin{example}	\label{ex:MMAX:mms}
	Consider the following example with 4 agents with additive valuations on 14 goods, where agent 1 and 2 have valuation shown in matrix $P$, while valuation of agent 3 and 4 are shown in matrix $Q$.
	The goods correspond to the non-zero elements in $P$ and $Q$.
	Here $\epsilon$ is a small number; say $\frac{1}{2^{10}}$ and $\tilde{\epsilon}$ is an even smaller number compared with $\epsilon$; say $\frac{1}{2^{100}}$.
	Partitioning the goods based on the rows gives an MMAX (and MMA1) allocation $A$.
	Note that in $A$, every agent $i$ gets a value of at least $1-\hat{\epsilon}$ and the least favorite good in $A_{-i}$ values at least $\epsilon^{4}$ to her. 
	Then by removing the least favorite good, the proportional value is at most $\frac{1}{3}\cdot (3+\tilde{\epsilon}-\epsilon^{4})$ (the total value of goods is $4$), which is no more than $1-\hat{\epsilon}$.
	Thus $v_{i}(A_{i})\geq \MMS_{i}(A\setminus\{e\},3)$ for any $e\in A_{-i}$ and $i\in N$.
	
	On the other hand, it is easy to check that $\MMS_i(A,4) = 1$ for all agent $i$.
	However, by enumerating all possible allocations, one can check that there does not exist any allocation that guarantees a bundle of value at least $1$ for every agent.
\end{example}
{	\[
	P=\left(
	\begin{array}{llll}
	\frac{7}{8} &  \epsilon^{4} & 0 &   \frac{1}{8} -\epsilon^{4} -\tilde{\epsilon}\\
	\epsilon^{3} &  \frac{3}{4} & -\epsilon^{3}+\epsilon^{2}  & \frac{1}{4}-\epsilon^{2} - \tilde{\epsilon}\\
	0 & -\epsilon^{4}+\epsilon  &  \frac{1}{2} &  \frac{1}{2}+\epsilon^{4}-\epsilon - \tilde{\epsilon}\\
	\frac{1}{8}-\epsilon^{3} & \frac{1}{4}-\epsilon & \frac{1}{2}+\epsilon^{3}-\epsilon^{2} &  \frac{1}{8}+\epsilon^{2} +\epsilon +3\tilde{\epsilon}
	\end{array}
	\right)
	\]
	
	\[
	Q=\left(
	\begin{array}{llll}
	\frac{7}{8} &  \epsilon^{4} & 0 &   \frac{1}{8} -\epsilon^{4}\\
	\epsilon^{3} &  \frac{3}{4} & -\epsilon^{3}+\epsilon^{2}  & \frac{1}{4}-\epsilon^{2}\\
	0 & -\epsilon^{4}+\epsilon  &  \frac{1}{2} &  \frac{1}{2}+\epsilon^{4}-\epsilon \\
	\frac{1}{8}-\epsilon^{3} - \tilde{\epsilon} & \frac{1}{4}-\epsilon - \tilde{\epsilon} & \frac{1}{2}+\epsilon^{3}-\epsilon^{2} -\tilde{\epsilon} &  \frac{1}{8}+\epsilon^{2}+\epsilon + 3\tilde{\epsilon}
	\end{array}
	\right)
	\]
}

Next we show the connection between MMAX and EFX when the valuations are binary additive.

\begin{lemma} \label{thm:MMAX:efx}
	EFX $\xRightarrow{BA}$ MMAX.
\end{lemma}
\begin{proof}
	Let $A$ be any EFX allocation.
	Consider any agent $i \in N$ and let $x = v_{i}(A_{i})$
	be the value of $i$'s bundle under the allocation $A$.
	Since $A$ is EFX, then for any $j\in N\setminus\{i\}$,
	one of the following statements holds:
	\begin{enumerate}
		\item $v_{i}(A_{j})=x+1$ and $|A_{j}|=x+1$ (i.e., agent $i$ has value 1 for every good in $A_j$);
		\item $v_{i}(A_{j})\leq x$ (i.e., $A_{j}$ may contain more than $x+1$ goods, but agent $i$ has value 1 for exact $x$ goods).
	\end{enumerate}
	This is because if both of them do not hold, then $v_{i}(A_{j})\geq x+1$ and $|A_{j}|>x+1$,
	which means agent $i$ has value $0$ for at least one good in $A_j$.
	Thus agent $i$ still envies agent $j$ after removing the good with value $0$ from $A_j$.
	
	Note that as long as there exists an agent $j\neq i$ such that $v_i(A_j) \leq x$, then $v_i(A_{-i}) \leq (x+1)(n-1)-1$. Hence
	\begin{equation*}
	\MMS_i(A_{-i},n-1)\leq x = v_i(A_i).
	\end{equation*}
	Otherwise we have $v_i(A_{-i}) = |A_{-i}| = (x+1)(n-1)$.
	Then by removing any good $e\in A_{-i}$, we have
	\begin{equation*}
	\MMS_i(A_{-i}\setminus\{e\},n-1) \leq x = v_i(A_i).
	\end{equation*}
	In both case, agent $i$ is satisfied with respect to MMAX. 
\end{proof}

However, when agents have general additive valuations, EFX does not imply MMAX.
More specifically, for the following example with additive valuations, an EFX allocation exists but is not MMAX.

\begin{example}
	Suppose $N=\{1,2,3\}$ and $M=\{1,2,\cdots,7\}$.
	Agent~1 has the additive valuation shown in Table \ref{table:example:MMAX:efx}.
	It is easy to check that $A_{1}=\{1\}$,  $A_{2}=\{2,4\}$, $A_{3}=\{3,5,6,7\}$ is an EFX allocation to agent 1.
	But this is not MMAX to her.
	Consider the following partition of $A_{-1}\setminus \{7\}$:
	$B_{2}=\{2,6\}$,  $B_{3}=\{3,4,5\}$. Then $v_{1}(B_{2})=v_{1}(B_{3})=1.2$,
	which are strictly larger than $v_1(A_1)$.
	\begin{table}[htbp]
		\begin{center}
			\begin{tabular}{|c|cccc|}
				\hline
				\mbox{Goods} & $1,2$ & $3$ & $4$ & $5,6,7$ \\
				\hline
				\mbox{Value} & 1 & 0.6 & 0.4 & 0.2 \\
				\hline
			\end{tabular}
		\end{center}
		\caption{EFX does not imply MMAX in general.}
		\label{table:example:MMAX:efx}
	\end{table}%
\end{example}

\paragraph{The Existence of MMAX Allocations.}

First, note that for two agents, MMAX is equivalent to EFX and thus exists for general subadditive valuations~\cite{plaut2018almost}.
Indeed, following exactly the same analysis as in Theorems~\ref{thm:MMA1:sym:agent} and~\ref{thm:MMA1:3agents},
we obtain the following result.


\begin{theorem}\label{thm:MMAX:sym:agent}
	When there are at most three agents having strictly increasing subadditive valuations, 
	or any number of agents having identical strictly increasing subadditive valuations,
	an MMAX allocation exists.
\end{theorem}


\section{Computing Allocations Efficiently} \label{sec:alg}

In this section, we provide a polynomial-time algorithm for computing an allocation such that when all agents have additive valuations
every agent is either $\frac{1}{2}$-MMA or MMAX.
Moreover, the allocation our algorithm outputs is $\frac{1}{2}$-EFX when all agents have subadditive valuations.

\paragraph{Outline of Algorithm.}
In each round (each while loop of Algorithm~\ref{alg:general}), we assign at most one good to each agent that is not envied by any other agent.
The assignment of goods to agents in each round is given by a maximum weight matching between the unenvied agents $L$ and the unallocated goods $R$, where the weight of every edge is given by the marginal increase in the value, after allocating the good to the corresponding agent.
If all edges between $L$ and $R$ have weight $0$, we assign goods to unenvied agents according to some maximum cardinality matching.
Observe that the set of unenvied agents may change after each round, e.g., after getting one more good, an unenvied agent may become envied by some other agents.
We repeat the procedure until all goods are allocated.
To make sure that there is always an agent not envied by any other agent, we use the approach (procedure $\mathcal{P}$ in the following) introduced in~\cite{lipton2004approximately}, which is now widely used in the resource allocation literature.

\begin{algorithm}[htbp]
	\caption{The Algorithm for Additive (or Subadditive) Valuations}
	\label{alg:general}
	\begin{algorithmic}[1]
		\STATE Initiate $L=N$ and $R=M$.
		\STATE Initiate $A_{i}=\emptyset$ for all $i\in N$.
		\WHILE{$R\neq \emptyset$}
		\STATE \label{step:matching} Compute a maximum weight matching $\mathcal{M}$ between $L$ and $R$, where the weight of edge between $i\in L$ and $j\in R$ is given by $v_i(A_i\cup\{j\}) - v_i(A_i)$. If all edges have weight $0$, then we compute a maximum cardinality matching $\mathcal{M}$ instead.
		\STATE For every edge $(i,j)\in \mathcal{M}$, allocate $j$ to $i$: $A_i = A_i \cup \{j\}$ and exclude $j$ from $R$: $R=R\setminus \{j\}$.\label{step:Alg1:1}
		\STATE \label{setp:deleteCycle:1}As long as there is an envy-cycle with respect to $A=(A_{i})_{i\in N}$, invoke procedure $\cP$ (to be described later).
		\STATE Update $A = (A_i)_{i\in N}$ to be the allocations after $\cP$.
		\STATE Update the set of agents not envied by any other agents: $L = \{ i\in N : \forall j\in N, v_j(A_j) \geq v_j(A_i) \}$. \label{step:alg:general:weight}
		\ENDWHILE
	\end{algorithmic}
\end{algorithm}

\paragraph{Procedure $\cP$.}
Given an allocation $A$, we define its corresponding {\em envy-graph} $\cG$ as follows.
Every agent is represented by a node in $\cG$ and there is a directed edge from node $i$ to node $j$ iff $v_{i}(A_{i}) < v_{i}(A_{j})$, i.e., $i$ is envies $j$.
A directed cycle in $\cG$ is called an {\em envy-cycle}.
Let $C=i_{1}\to i_{2}\to\cdots\to i_{t} \to i_{1}$ be such a cycle.
If we reallocate $A_{i_{k+1}}$ to agent $i_{k}$ for all $k\in [t-1]$, and reallocate $A_{i_1}$ to agent $i_t$, the number of edges of $\cG$ will be strictly decreased.
Thus, by repeatedly using this procedure, we eventually get another allocation whose envy-graph is acyclic.
It is shown in \cite{lipton2004approximately} that $\cP$ can be done in time $O(mn^3)$.	

Next, we prove the following main result of this section.

\begin{theorem} \label{thm:alg2:efx}
	Algorithm \ref{alg:general} computes an allocation in polynomial time that is (1)
	$\frac{1}{2}$-EFX and EF1 when all agents have subadditive valuations; 
	(2) either $\frac{1}{2}$-MMA or exact MMAX for each agent when the valuations are additive.
\end{theorem}
\begin{proof}
	We first show that Algorithm \ref{alg:general} is well defined.
	In each round, there must be at least one agent who is not envied by anybody,
	since after Step \ref{setp:deleteCycle:1}, the envy-graph is acyclic.
	Thus there will be at least one agent whose in-degree is $0$, i.e., $L\neq \emptyset$.
	Moreover, since at least one good will be allocated in each round (each while loop), Algorithm~\ref{alg:general} terminates in at most $m$ rounds.
	Given that each round (including procedure $\cP$, the computation of maximum weight matching, and the updates of edge weights and unenvied agents $L$) can be done in polynomial time, Algorithm~\ref{alg:general} runs in polynomial time.
	
	Fix any agent $i$. We classify other agents into three sets:
	\begin{compactitem}
		\item agents not envied by $i$: $N_1 = \{ j\in N\setminus\{i\} : v_i(A_i) \geq v_i(A_j) \}$;
		\item agents envied by $i$ that receive only one good: $N_2 = \{ j\in N\setminus\{i\} : v_i(A_i) < v_i(A_j) \text{ and } |A_j|=1 \}$;
		\item agents envied by $i$ that receive more than one good: $N_3 = \{ j\in N\setminus\{i\} : v_i(A_i) < v_i(A_j) \text{ and } |A_j|\geq 2 \}$.
	\end{compactitem}
	
	By definition we have $N_1\cup N_2\cup N_3 = N\setminus \{i\}$.
	
	We first show that the final allocation $A$ is $\frac{1}{2}$-EFX and EF1 when the valuations are subadditive.
	Note that it suffices to consider $N_3$, as the other agents are either not envied by $i$, or allocated only one good.

	Fix any agent $j\in N_3$.
	Let $e_j$ be the last good allocated to $j$.
	As we only allocate goods to unenvied agents, we have $v_i(A_i) \geq v_i(A_j\setminus \{e_j\})$, which implies that agent $i$ is satisfied with respect to EF1.
	On the other hand, we also have $v_i(A_i) \geq v_i(e_j)$, since otherwise in the first round, matching $e_j$ (which is not allocated yet) with $i$ (which is not envied) gives a matching with strictly larger weight.
	Combing the two inequalities and by subadditivity of valuation $v_i$, we have
	\begin{equation*}
		2\cdot v_i(A_i) \geq v_i(A_j\setminus \{e_j\}) + v_i(e_j) \geq v_i(A_j),
	\end{equation*}
	which implies that agent $i$ is satisfied with respect to $\frac{1}{2}$-EF, and $\frac{1}{2}$-EFX.
	
	Finally, we show that when the valuations are additive, agent $i$ is also satisfied with respect to $\frac{1}{2}$-MMA or MMAX.
	
	The case when $N_2 = N\setminus \{i\}$, i.e., all other agents receive one good and is envied by agent $i$, is trivial: by removing any good $e$ from $A_{-i}$, there is always some agent with no good.
	Thus we have
	\begin{equation*}
	\MMS_i(A_{-i}\setminus \{e\},n-1) = 0 \leq v_i(A_i),
	\end{equation*}
	which implies that agent $i$ is also satisfied with respect to MMAX.
	
	Now suppose that $N_2 \neq N\setminus \{i\}$. In other words, we have $N_1\cup N_3 \neq \emptyset$, and $n-1-|N_2|\geq 1$. We prove that in this case we have $\MMS_i(A_{-i},n-1) \leq 2\cdot v_i(A_i)$, which implies that agent $i$ is satisfied with respect to $\frac{1}{2}$-MMA. 
	
	We first prove that $\MMS_{i}(A_{-i}, n-1)$ is at most the average value of bundles allocated to agents in $N_1\cup N_3$:
	\begin{equation*}
		\MMS_{i}(A_{-i}, n-1) \leq \frac{\sum_{j\in N_1} v_{i}(A_{j}) + \sum_{j\in N_3}v_{i}(A_j)}{n-1-|N_2|}.
	\end{equation*}
	
	Let $\tau = \MMS_{i}(A_{-i}, n-1)$.
	Observe that if we have $v_i(e)\leq \tau$ for every $e\in A_{-i}$, then 
	\begin{align*}
	& \sum_{j\in N_1} v_i(A_j) + \sum_{j\in N_3}v_i(A_j) = v_i(M) - \sum_{j\in N_2} v_i(A_j) \\
	\geq & v_i(M) - |N_2|\cdot \tau \geq (n-1-|N_2|)\cdot \MMS_{i}(A_{-i}, n-1).
	\end{align*}
	Indeed, the assumption is w.l.o.g., since otherwise we can reset every $v_i(e)$ to $\min\{ \tau, v_i(e) \}$, which does not change $\MMS_{i}(A_{-i}, n-1)$.
	The reason is, in any partitioning of $A_{-i}$ into $n-1$ bundles, the bundle with minimum value (whose value is at most $\tau = \MMS_i(A_{-i}, n-1)$) is not touched.
	
	As we have shown before (for subadditive valuations), for every $j\in N_3$, we have $v_i(A_j)\leq 2\cdot v_i(A_i)$.
	Since we have $v_i(A_j) \leq v_i(A_i)$ for all $j\in N_1$, we conclude that
	\begin{equation*}
	\MMS_i(A_{-i},n-1) \leq 2\cdot v_i(A_i),
	\end{equation*}
	which implies that agent $i$ is satisfied with respect to $\frac{1}{2}$-MMA.
\end{proof}


\paragraph{Remark.}
For additive valuations, 
by Theorem~\ref{thm:alg2:efx}, the algorithm returns a $\frac{1}{2}$-MMAX allocation in general.
Actually, by the proof of Theorem~\ref{thm:alg2:efx}, it is not difficult to see that the returned allocation is
\begin{compactitem}
	\item $\frac{1}{2}$-MMA if at least two agents receive more than one good;
	\item $\frac{1}{2}$-MMAX if exactly one agent receives more than one good;
	\item MMAX, otherwise.
\end{compactitem}

Similarly, when the valuations are subadditive, we are able to show that the allocation is 
$\frac{1}{2}$-EF, when at least two agents receive more than one good;
$\frac{1}{2}$-EFX, when exactly one agent receives more than one good;
and EFX, otherwise.

\section{Conclusion} \label{sec:discussion}
%
%
%

In this paper, we introduce a novel fairness notions MMA and its two relaxations MMA1 and MMAX in an unaware environment.
We study their connections with other fairness notations, and propose an efficient algorithm for computing allocations that is (approximately) MMA and MMAX.

We leave the existence of MMA1 and MMAX allocations for the broader class of valuations as a future direction.
Another promising direction would be to extend our work to other preference representations, including ordinal preferences, or to chores instead of goods. 
Finally, it will be interesting to obtain analogues concepts for asymmetric agents, where agents may have different entitlements or shares.

\section*{Acknowledgments}
The authors thank Haris Aziz and David Parkes or reading a draft of this paper and for helpful discussions.
This work is partially supported by NSF CAREER Award No. 1553385.

\bibliographystyle{apalike}
\bibliography{mma}

\end{document}